\documentclass[11pt]{amsart}
\usepackage{latexsym,amssymb,amsmath,youngtab}
\usepackage{youngtab}
\usepackage{epsfig}
\usepackage{amsmath,amsthm,amssymb,amscd,MnSymbol}
\usepackage[mathscr]{eucal}
\usepackage{verbatim, hyperref}
\usepackage{color}
\newcommand*{\textlabel}[2]{%
  \edef\@currentlabel{#1}
  \phantomsection
  #1\label{#2}
}

\newtheoremstyle{custom}
  {3pt}
  {3pt}
  {\slshape}
  {}
  {\bfseries}
  {.}
  { }
   {}
\theoremstyle{custom}
\newtheorem{theorem}{Theorem}[section]
\newtheorem{proposition}[theorem]{Proposition}

\newtheorem{proposition/definition}[theorem]{Proposition/Definition}
\newtheorem{lemma}[theorem]{Lemma}
\newtheorem{corollary}[theorem]{Corollary}

\theoremstyle{definition}
\newtheorem{definition}[theorem]{Definition}

\newtheorem{example}[theorem]{Example}

\theoremstyle{remark}
\newtheorem{remark}[theorem]{Remark}

\newtheoremstyle{exercise}
  {3pt}
  {6pt}
  {}
  {}
  {\bfseries}
  {:}
  { }
   {}
\theoremstyle{exercise}
\newtheorem{exercise}[theorem]{Exercise}
\newtheoremstyle{exercises}
  {3pt}
  {6pt}
  {}
  {}
  {\bfseries}
  {:}
  {\newline}
   {}

\theoremstyle{exercise}
\newtheorem{exercises}[theorem]{Exercises}

\input epsf
\def\boxit#1{\vbox{\hrule height1pt\hbox{\vrule width1pt\kern3pt
  \vbox{\kern3pt#1\kern3pt}\kern3pt\vrule width1pt}\hrule height1pt}}

\def\trank{\text{rank}}

\def\bv{\bold v}

\def\BC{\mathbb C}

\def\BP{\mathbb P}

\def\tdim{{\rm dim}}

\def\hd{,...,}
\def\ww{\wedge}
\def\upperp{{}^\perp}

\def\cS{{\mathcal S}}

\def\11{\mathbf 1}
\def\PP{\mathbb P}

\def\l{\lambda}

\def\o{\omega}

\def\b{\beta}

\def\s{\sigma}

\def\ot{{\mathord{ \otimes } }}

\def\otc{{\mathord{\otimes\cdots\otimes}\;}}

\def\ra{{\mathord{\;\rightarrow\;}}}

\def\dim{{\rm dim}\;}
\def\La#1{\Lambda^{#1}}

\def\EE#1#2{E^{#1}_{#2}}

\def\s{\sigma}

\def\b{\beta}

\def\l{\lambda}

\def\BP{\mathbb  P}
\def\BC{\mathbb  C}

\def\hd{, \hdots ,}

\def\La#1{\Lambda^{#1}}

\def\ur{\underline {\bold R}}

\def\ra{\rightarrow}

\def\tdim{\operatorname{dim}}
\def\tker{\operatorname{ker}}
\def\tlim{\lim}

\def\tmin{\operatorname{min}}

\def\trank{\operatorname{rank}}

\def\upperp{{}^{\perp}}

\def\ww{\wedge}

\def\bbb{{\bold{b}}}

\def\be{\begin{equation}}
\def\ene{\end{equation}}
\def\aaa{{\bold {a}}}
\def\bbb{{\bold {b}}}
\def\ccc{{\bold {c}}}

\DeclareMathOperator{\tlog}{log}
\DeclareMathOperator\tspan{span}

\def\p{{\bold P}}

\def\rank{\operatorname{rank}}

\def\tspan{{\rm span}}

\newcommand{\Id}{\operatorname{Id}}

\def\EE{\mathcal{E}}

\def\Mn{M_{\langle \nnn \rangle}}

\def\Mnw{M_{\langle \nnn,\nnn,  \bw\rangle}}
\def\Mtwo{M_{\langle 2\rangle}}

\def\trank{{\mathrm {rank}}}

\def\aaa{{\bold a}}\def\bbb{{\bold b}}\def\ccc{{\bold c}}

\def\uuu{\bold u}

\newcommand{\GL}{\operatorname{GL}}

\def\mmm{\bold m}\def\nnn{\bold n}

\def\bv{\bold v}\def\bw{\bold w}

\def\BRe{B}

 \def\Mnwredl{M_{\langle \nnn,\nnn,\bw \rangle}^{\l}} \def\Mnwredlp{M_{\langle \nnn,\nnn,\bw \rangle}^{\l'}}

\begin{document}

\author{J.M. Landsberg}
\address{
Department of Mathematics\\
Texas A\&M University\\
Mailstop 3368\\
College Station, TX 77843-3368, USA}
\email{jml@math.tamu.edu}
\author{Mateusz Micha{\l}ek}
\address{
Polish Academy of Sciences\\
         ul. \'Sniadeckich 8\\
         00-956 Warsaw\\
         Poland}
\email{wajcha2@poczta.onet.pl}
\thanks{Landsberg    supported by   NSF grant  DMS-1405348. Michalek was supported by Iuventus Plus grant 0301/IP3/2015/73 of the Polish Ministry of Science.}
\title[The border rank of matrix multiplication]{A $2\nnn^2-\tlog_2(\nnn)-1$ 
lower bound for the border rank of matrix multiplication}
\maketitle
\begin{abstract}
Let $\Mn\in \BC^{\nnn^2}\ot\BC^{\nnn^2}\ot\BC^{\nnn^2}$ denote the
matrix multiplication tensor for $\nnn\times\nnn$ matrices.
We use the border substitution method \cite{MR3025382,2016arXiv160604253B,2015arXiv150403732L} combined with
Koszul flattenings \cite{MR3376667} to prove
the border rank lower bound $\ur(M_{\langle \nnn,\nnn,\bw\rangle})\geq 2\nnn^2 - \lceil \tlog_2(\nnn)\rceil -1$.  
\end{abstract}

\section{Introduction}

Let $A,B,C,U,V,W$ be vector spaces of dimensions $\aaa,\bbb,\ccc,\uuu,\bv,\bw$.
The matrix multiplication tensor
$M_{\langle \uuu,\bv,\bw\rangle}\in (U^*\ot V)\ot (V^*\ot W)\ot (W^*\ot U)$ is given in 
coordinates by 
$$
M_{\langle \uuu,\bv,\bw\rangle}
=\sum_{i=1}^{\uuu}\sum_{j=1}^{\bv}\sum_{k=1}^{\bw}
x^i_j\ot y^j_k\ot z^k_i.
$$
Ever since Strassens' discovery \cite{Strassen493} that the
standard algorithm for multiplying matrices is not optimal, the matrix multiplication
tensor has been a central object of study.
We write $\Mn=M_{\langle \nnn,\nnn,\nnn\rangle}$.

Let $T\in A\ot B\ot C$ be a tensor. The {\it rank} of $T$ is the smallest $r$ such that 
$T$ may be written as a sum of $r$ rank one tensors (tensors of the form $a\ot b\ot c$ for $a\in A,b\in B,c\in C$).
The {\it border rank} of $T$ is the smallest $r$ such that $T$ may be written as a limit of rank $r$ tensors.
We write $\ur(T)=r$.
Border rank is a basic measure of the complexity of a tensor.
For example, the exponent of matrix multiplication, the smallest $\o$ such that $\nnn\times \nnn$ matrix multiplication
can be computed with $O(\nnn^{\o})$ arithmetic operations, satisfies 
  $ \o= \underline{\tlim}_{\nnn\ra\infty} \tlog_{\nnn}(\ur(\Mn))$.  All modern upper and lower bounds for the complexity of matrix multiplication rely implicitly or explicitly on border rank.  
  Strassen showed $\ur(\Mn)\geq \frac{3\nnn^2}2$ \cite{Strassen505}
and Lickteig improved this to $\ur(\Mn)\geq \frac{3\nnn^2}2+\frac\nnn 2 -1$ \cite{MR86c:68040}.
After that, progress stalled for nearly thirty years (other than showing $\ur(\Mtwo)=7$ \cite{MR2188132}), until
in 2012 the first author and Ottaviani  showed $\ur(\Mn)\geq 2\nnn^2-\nnn$ \cite{MR3376667}.
In 2016 we improved this to $\ur(\Mn)\geq 2\nnn^2-\nnn+1$ \cite{2016arXiv160108229L}. More important than the result in  \cite{2016arXiv160108229L} was
the method of proof - a border rank version of the {\it substitution method} \cite{MR3025382,2016arXiv160604253B,2015arXiv150403732L}.
We use this method in  a more refined way to prove:

\begin{theorem}\label{mmultthm}  Let $0<m<\nnn$. 
Then   
$$
\ur(\Mnw)\geq 2\nnn\bw -\bw+m-
\lfloor \frac{\bw\binom{\nnn-1+m}{ m-1}}{\binom{2\nnn-2}{\nnn-1}}\rfloor.
$$
In particular,  taking $\bw=\nnn$ and  $m=\nnn- \lceil\tlog_2(\nnn)\rceil -1$, 
$$
\ur(\Mn)\geq 2\nnn^2-  \lceil\tlog_2(\nnn)\rceil -1.
$$
\end{theorem}

 \medskip
As can be seen in the proof, one can get a slightly better lower bound. 
Here are a few cases with optimal $m$ and the improvement over the previous bound:
\begin{align*}
\nnn\ \ \   &\ur(\Mn)\geq \ \ \ &{\rm improvement\ over\ }2\nnn^2-\nnn +1\\
4\ \ \   & 29 & 0\\
5\ \ \   & 47 & 1\\
6\ \ \   & 69 & 2\\
7\ \ \   & 95 & 3\\
8\ \ \   & 122 & 3\\
9\ \ \   & 158 & 4\\
10\ \ \   & 196 & 6\\
100\ \ \   & 19,992 & 92\\
1000\ \ \   & 1,999,989 & 989\\
10,000\ \ \ & 199,999,985 & 9985
\end{align*}

\medskip

The substitution and border substitution methods na\"\i vely could be used
to prove rank and border rank lower bounds up to   $3\mmm-3$ for tensors in $\BC^\mmm\ot \BC^\mmm\ot\BC^\mmm$.
We show this is not quite possible  for border rank. We define a variety $X(\aaa',\bbb',\ccc')\subset \BP (A\ot B\ot C)$
that corresponds to tensors where the border substitution method 
fails
to provide lower bounds beyond $\aaa+\bbb+\ccc-\aaa'-\bbb'-\ccc'$. 
More
precisely, $X(\aaa',\bbb',\ccc')$ is the variety of {\it $(\aaa',\bbb',\ccc')$-compressible tensors}, those for
which there exists $A'\subset A^*$, $B'\subset B^*$, $C'\subset C^*$, 
respectively
of dimensions $\aaa',\bbb',\ccc'$, such that $T$,
considered as a linear form on  $A^*\ot B^*\ot C^*$, satisfies $T\mid_{A'\ot B'\ot C'}=0$.
We show:

\begin{proposition}\label{prop:fill}
The set  $X(\aaa',\bbb',\ccc')\subseteq \BP
(A\ot B\ot C)$ is Zariski closed. If
\be\label{surjdim}\aaa \aaa ' + \bbb\bbb'+\ccc\ccc'
<(\aaa')^2+(\bbb')^2+(\ccc')^2 + \aaa'\bbb'\ccc'
\ene
then $X(\aaa',\bbb',\ccc')\subsetneq\BP (A\ot B\ot C)$.
In particular, in the range where \eqref{surjdim} holds,   the   substitution 
methods  may be used to prove nontrivial
lower bounds for 
border rank.
\end{proposition}

The proof and examples  show that beyond this bound one expects $X(\aaa',\bbb',\ccc')=\BP (A\ot B\ot C)$, so
that the method cannot be used. 

Note that
if  $\ur(T)\leq \aaa+\bbb+\ccc-(\aaa'+\bbb'+\ccc')$ then there exists $A'\subset A^*,B'\subset B^*,C'\subset C^*$  such that $T|_{A'\ot B'\ot C'}=0$.
Let $\s_r(Seg(\BP A\times \BP B\times \BP C))\subset \BP (A\ot B\ot C)$
denote the variety of tensors of border rank at most $r$, called
the $r$-th secant variety of the Segre variety.
The above remark  may be restated as 
\begin{proposition}
  $$\sigma_{\aaa+\bbb+\ccc-(\aaa'+\bbb'+\ccc')}Seg(\BP A\times \BP B\times \BP C)\subset X(\aaa',\bbb',\ccc').$$
\end{proposition}

We expect the inequality in Proposition \ref{prop:fill} to be sharp or nearly
so. For tensors in $\BC^{\mmm}\ot \BC^{\mmm}\ot \BC^{\mmm}$ the
limit of this method alone would be a border rank lower bound of
$3(m-\sqrt{3m +\frac 94}+\frac 32)$.
    However, it is unlikely the method
alone could attain such a bound due to technical difficulties in proving an explicit tensor does not belong to $X(\aaa',\bbb',\ccc')$. 

The state of the art for matrix multiplication is such that 
on one hand, for
upper bounds on the exponent  there does not appear
to be a viable path proposed for proving the exponent is less than
2.3, but on the other,   none of
the existing  techniques appear to be able to prove a border rank lower
bound
of $2n^2$ for matrix multiplication. 
 
\subsection{Acknowledgements} We thank Jason Starr and Math Overflow for help with Example \ref{jasonex}.  Michalek is a member of the AGATES group and is supported by the Foundation for Polish Science (FNP).

\section{Preliminaries}

Let $A=U^*\ot V$, $B=V^*\ot W$, $C=U\ot W^*$.
For $v\in V$, we write $\hat v\subset V$ for the line it determines
and $[v]\in \BP V$ for the  corresponding  point in projective space. 
 
\begin{definition}
For a  tensor $T\in V_1\ot\dots\ot V_n$, and $U\subset V_1$,  let $T/U \in (V_1/U)\ot V_2\ot\dots\ot V_n$
denote $T\mid_{U\upperp \ot V_2^*\otc V_{n}^*}$, where we consider $T$ as a linear form on $V_1^*\otc V_n^*$.
Define 
$$\BRe_k(T):=\{[v]\in \PP V_1\mid   \ur(T/\hat v)\leq k\}.
$$
\end{definition}
\begin{lemma}\label{lem:closed}
Let $T\in V_1\ot\dots\ot V_n$ be a tensor, let  $G_T\subset \GL(V_1)\times\dots\times \GL(V_n)$ denote its stabilizer
and let $G_1\subset GL(V_1)$ denote its projection to $GL(V_1)$. The set $\BRe_k(T)$ is:
\begin{enumerate}
\item Zariski closed,
\item a $G_1$-variety.
 \end{enumerate}
\end{lemma}
\begin{proof}
(1)
Let $\mathcal{L}$ be the total space of the quotient bundle over $\PP V_1 $ tensored
with $V_2\otimes\dots\otimes V_n$, i.e.~the fiber over $[v]$ is $(V_1/v)\ot V_2\ot\dots\ot V_n$. 
We have a natural section $s:\PP V_1 \rightarrow \mathcal{L}$ defined by $s([v]):=T/v$.
Let  $X\subset \mathcal{L}$ 
denote the sub-bundle whose 
fiber over $[v]\in  \PP V_1 $  is the locus of tensors of border rank at most $k$ in $(V_1/v)\ot V_2\ot\dots\ot V_n$. 
The set $\BRe_k(T)$ is the projection to $\PP V_1 $ of the intersection of the  image of the  section $s$ and $X$.

(2) Let $g=(g_1,\dots,g_n)\in G_T$. Then  $\ur(T/v)=\ur(gT/g_1v)=\ur(T/g_1v)$.
\end{proof}

A tensor $T\in A\ot B\ot C$ is {\it $A$-concise} if it is not
contained in any $A'\ot B\ot C$ where $A' \subsetneq  A$.

\begin{proposition}\label{subsprop} \cite{2016arXiv160604253B,2015arXiv150403732L}
Let $T\in A\ot B\ot C$ be $A$-concise. Fix $\aaa'\leq \aaa$.
Then 
$$
\ur(T)\geq \tmin_{A'\in G(\aaa', A^*)}\ur(T|_{A'\ot B^*\ot C^*})+ (\aaa-\aaa').
$$
\end{proposition}

\begin{remark} The situation for rank is slightly better than  
for border rank in that one
can choose $A'$ at the price of making a suitable modification of $T$, see \cite{MR3025382,2015arXiv150403732L}.
\end{remark}

We will use the Koszul flattening of \cite{MR3376667}: for $T\in A\ot B\ot C$, define
\be\label{kflat}
T_A^{\ww p}: B^*\ot \La p A\ra \La{p+1}A\ot C
\ene
  by first taking $T_B\ot \Id_{\La p}A: B^*\ot \La p A\ra \La p A\ot A\ot C$, and then projecting to $\La{p+1}A\ot C$.
If $\{ a_i\},\{ b_j\},\{ c_k\}$ are bases of $A,B,C$ and 
  $T=\sum_{i,j,k}t^{ijk}a_i\ot b_j\ot c_k$,
then
\be\label{kflatany}
T_{A}^{\ww p}(\b\ot f_1\ww\cdots \ww f_p)=\sum_{i,j,k} t^{ijk}\b(b_j)a_i\ww f_1\ww\cdots \ww f_p\ot c_k.
\ene
We have \cite{MR3376667}:
\be\label{yfinequal}
\ur(T)\geq \frac{\trank(T_{A}^{\ww p})}{\binom{\aaa-1}{p}}. 
\ene
In practice the map $T_{A}^{\ww p}$ is used after  specializing
$T$ to a subspace of $A$ of dimension $2p+1$ to get a potential   $\frac{2p+1}{p+1}\bbb$ border rank
lower bound.

\section{Proof of Theorem \ref{mmultthm}}
We first observe that the \lq\lq In particular\rq\rq\ assertion follows from the main assertion because,
taking $m=\nnn-c$, we want $c$ such that 
$$
\frac{\nnn\binom{2\nnn-1-c}{\nnn}}{ \binom{2\nnn-2}{\nnn-1}}<1
$$
This ratio is
$$
\frac{ (\nnn-1)\cdots (\nnn-c)}{(2\nnn-2)(2\nnn-3)\cdots (2\nnn-c)}=
\frac{ \nnn-c}{2^{c-1}}\frac{\nnn-1}{\nnn-\frac 22}\frac{\nnn-2}
{\nnn-\frac 32}\frac{\nnn-3}{\nnn-\frac 42}\cdots \frac{\nnn-c+1}{\nnn-\frac c2}
$$
so if $c-1 \geq \tlog_2(\nnn)$ it is less than one.

\medskip

For the rest of the proof, we first introduce notation: 
for a Young diagram $\l$, we picture  it Russian style,
as we think of it as representing entries in the  south-west corner of an $\nnn\times \nnn$ matrix.
More precisely  
for $(i,j)\in\lambda$ we number the boxes of $\lambda$ by pairs (row,column)
however we number the rows starting from $\nnn$, i.e.~$i=\nnn$ is the first row.
For example
$$
\young(xy,z,w)
$$
is labeled $x=(\nnn,1),y=(\nnn,2),z=(\nnn-1,1),w=(\nnn-2,1)$.
Let $U_{\l}:=\tspan\{ u^i\ot v_j \mid (i,j)\in \l\}$ and write
$\Mnwredl:=M_{\langle \nnn,\nnn,\bw\rangle}/U_{\l}$.

The proof consists of two parts. In the first,   we prove by induction
on $k$ that for any $k<\nnn$ there exists a Young diagram $\l$
with $k$ boxes such that
$\ur(\Mnwredl)\leq \ur(\Mnw)-k$.  

In the second part we estimate $\ur(\Mnwredl)$ for any $\lambda$  by reducing to the case when $\lambda$ has just one row (or column).

\bigskip

{\bf Part 1)} First step: $k=1$. By Proposition  \ref{subsprop}
there exists   $a\in \BRe_{\ur(\Mnw)-1}(\Mnw)$ such that the reduced tensor drops border rank. 
The group $\GL(U)\times  \GL(V)\times \GL(W)$ stabilizes $\Mnw$. 
By Lemma \ref{lem:closed}  with  $G_1=\GL(U)\times \GL(V)$,
we may act on  $a$ and pass to the limit. Hence, we may first reduce the rank of $a$ to $1$ and then make it equal $u^\nnn\ot v_1$. 

Second step: We assume that $\ur(\Mnwredlp)\leq \ur(\Mnw)-k+1$, where $\lambda'$ has $k-1$ parts. Again by Proposition  \ref{subsprop}
there exists $a\in \BRe_{\ur(\Mnw)-k}(\Mnwredlp)$ 
such that when we reduce by it the border rank drops. We no longer have the full action of $\GL(U)\times  \GL(V)$. However, the  product of Borel groups that stabilize the flags  induced
by $\l'$
stabilizes $\Mnwredlp$. By the torus action and Lemma \ref{lem:closed} we may assume that $a$ has just one nonzero entry outside of $\lambda$. Further, using the Borel action we can move the entry  south-west  to obtain  the desired Young diagram $\lambda$. 

\bigskip

{\bf Part 2)}
We use \eqref{kflat} and recall that for the matrix multiplication operator, the Koszul flattening factors as
$\Mnw=M_{\langle \nnn, \nnn, 1\rangle}\ot \Id_W$, so
we apply the  Koszul  flattening to $M_{\langle \nnn, \nnn, 1\rangle}\in (U^*\ot V)\ot V^*\ot U$, where $\uuu=\bv=\nnn$.
We need to show that for all $\l$ of size $m$, 
$$
\ur(M_{\langle \nnn, \nnn, 1\rangle}^{\lambda})\geq 2\nnn-1-
\frac{\binom{\nnn-1+m}{m-1}}
{\binom{2\nnn-1}{\nnn-1}}.
$$
We will accomplish this  by projecting  to a suitable $p_{\tilde A}:A\rightarrow \tilde A$  of dimension $2\nnn-1$,
such that 
$$\trank ([p_{\tilde A}(M_{\langle \nnn, \nnn, 1\rangle}^{\lambda}) )]^{\ww \nnn-1}_{\tilde A}
 \geq
\binom{2\nnn-1}{\nnn-1}\nnn -  
 \binom{\nnn-1+m}{m-1} 
$$
and then apply \eqref{yfinequal}.  By our choice of   basis we may   consider $M_{\langle \nnn, \nnn, 1\rangle}^{\lambda}\in (A/U_\lambda)\ot B\ot C$ 
in $A\ot B\ot C$, with specific coordinates equal to $0$.
We need to show
$$
\tdim \tker ([p_{\tilde A}(M_{\langle \nnn, \nnn, 1\rangle})^{\lambda} ]^{\ww \nnn-1}_{\tilde A})\leq 
\binom{\nnn-1+m}{m-1}.
$$

  Consider the map $\phi:A\ra \BC^{2\nnn-1}$ given by $u^i\ot v_j\mapsto e_{i+j-1}$. The rank of the reduced Young flattening 
  $\La {n-1} \BC^{2\nnn-1}\ot V\ra \La{n}\BC^{2\nnn-1}\ot U$ could only go down. However, for $M_{\langle \nnn,\nnn,1\rangle}$,
  as was shown in \cite{MR3376667,2016arXiv160108229L},
  the new  map is surjective. We  recall the argument
  from \cite{2016arXiv160108229L}, as a similar argument will finish the proof. 

Write $e_S=e_{s_1}\ww\cdots \ww e_{s_{\nnn-1}}$, where $S\subset [2\nnn-1]$ has cardinality $\nnn-1$.
For  $1\leq\eta\leq \nnn$ the reduced Koszul  flattening is given by: 
$$
e_S\ot v_\eta\mapsto \sum_{j=1}^\nnn  \phi(u^j\ot v_\eta)\ww e_S\ot u_j 
= \sum_{j=1}^\nnn e_{j+\eta-1}\ww e_S\ot u_j.$$

We index a basis of the source by pairs $(S,k)$, with $k\in [\nnn]$, and the target by $(P,l)$ where
$P\subset [2\nnn-1]$ has cardinality $\nnn $ and $l\in[\nnn]$.
Define an order   on the target basis vectors as follows: For $(P_1,l_1)$ and $(P_2,l_2)$, set $l=\tmin\{ l_1,l_2\}$, and declare
$(P_1,l_1)<(P_2,l_2)$ if and only if
\begin{enumerate}
\item In lexicographic order,  the set of $l$ minimal elements of $P_1$ is strictly after  the set of $l$ minimal elements of $P_2$ (i.e.~the smallest element of $P_2$ is smaller than the smallest of $P_1$ or they are equal and the second smallest of $P_2$ is smaller or equal etc.~up to $l$-th),  or
\item the $l$ minimal elements in $P_1$ and $P_2$ are the same, and $l_1<l_2$. 
\end{enumerate}
In \cite{2016arXiv160108229L} we showed that when  one   orders the basis as above,   the reduced Koszul  flattening for $\Mn$ has an upper 
triangular structure. More explicitly, let $P=(p_1\hd p_{\nnn})$ with $p_i<p_{i+1}$. Identifying basis vectors with their indices,
the image of $(P\backslash \{ p_l\}, 1+p_l-l)$ is $\pm (P,l)$ plus smaller terms in the order. 
The crucial part is to control how the projection   of $\Mnwredl$ 
to the complement of $ u^j\ot v_{n+1-i}$  effects the reduced Koszul  flattening. We  determine  the number of additional zeros on the diagonal. 
Note that $(P,l)$ will not appear as the leading term any more if and only if $l=j$ and $n+1-i+j-1=p_l$. 
Hence, the number of additional zeros on the diagonal equals the number of $n$ element subsets of $[2n-1]$
that have the $j$-th entry equal to $n-i+j$, which is ${{n-i+j-1}\choose{j-1}}{{n+i-j-1}\choose{i-1}}:=g(i,j)$. So
it is enough to prove that $\sum_{(i,j)\in\lambda}g(i,j)\leq \binom{\nnn-1+m}{ m-1}$. Note that
$\sum_{i=1}^m g(i,1)=\sum_{j=1}^m g(1,j)=\binom{\nnn-1+m}{ m-1}$. Thus we have to prove that the Young diagram that maximizes $f_\lambda:= \sum_{(i,j)\in\lambda}g(i,j)$ has one row or column.
We prove it inductively on the size of $\lambda$, the case $|\lambda|=1$ being trivial. 

Suppose now that $\lambda=\lambda'+(i,j)$. By induction it is sufficient to show that:
\begin{equation}\label{ineq:main}
g(1,ij)={{\nnn-1+ij-1}\choose{ij-1}}\geq {{\nnn-j+i-1}\choose{i-1}}{{\nnn-i+j-1}\choose{j-1}}=g(i,j),
\end{equation}
where $\nnn>ij$.  
Without loss of generality we may assume $2\leq i\leq j$. For $j=2,3$ the inequality is straightforward to check, so we assume $j\geq 4$.
We prove the inequality \ref{ineq:main} by induction on $\nnn$.
For $\nnn=ij$ the inequality follows from the combinatorial interpretation of binomial coefficients and the fact that the middle one is the largest.

We have ${{\nnn+1-1+ij-1}\choose{ij-1}}={{\nnn-1+ij-1}\choose{ij-1}}\frac{\nnn-1+ij}{\nnn}$, ${{\nnn+1-j+i-1}\choose{i-1}}={{\nnn-j+i-1}\choose{i-1}}\frac{\nnn-j+i}{n-j+1}$ and ${{\nnn+1-i+j-1}\choose{j-1}}={{\nnn-i+j-1}\choose{j-1}}\frac{\nnn-i+j}{\nnn-i+1}$.
By induction it is enough to prove that:
\begin{equation}\label{eq:want}\frac{n-1+ij}{n}\geq \frac{n-j+i}{n-j+1}\frac{n-i+j}{n-i+1}.\end{equation}
This is equivalent to:
$$ij-1\geq \frac{n(i-1)}{n-j+1}+\frac{n(j-1)}{n-i+1}+\frac{n(i-1)(j-1)}{(n-j+1)(n-i+1)}.$$
As the left hand side  is independent from $n$ and each fraction on the right hand side decreases with growing $n$, we may set $n=ij$ in inequality \ref{eq:want}. Thus it is enough to prove:
$$2-\frac{1}{ij}\geq (1+\frac{i-1}{ij-j+1})(1+\frac{j-1}{ij-i+1}).$$
Then the inequality is straightforward to check for $i=2$, so we assume $i\geq 3$. Then:
$$(1+\frac{i-1}{ij-j+1})(1+\frac{j-1}{ij-i+1})\leq
(1+\frac{j-1}{j^2-j+1})(1+\frac{j-1}{3j-2})\leq\frac{16}{13}\cdot\frac{4}{3}=\frac{64}{39}.
$$
However,
$$\frac{64}{39}\leq 2-\frac{1}{12}\leq 2-\frac{1}{3j}\leq 2-\frac{1}{ij},$$
which finishes the proof.

\begin{remark} Note that we made two kinds of restrictions:
\begin{enumerate}
\item projecting $A$ to $A/U_{\lambda}$ and
\item projecting $A/U_{\lambda}$ to $\tilde A$.
\end{enumerate}
The first one corresponds to deleting rows (specified by $\lambda$) in the matrix representation of $M_{\langle \nnn, \nnn, 1\rangle}$. The second one takes $2\nnn-1$ linear combinations of rows as explained below. 
 
Since linear  projections  commute, one might
try to first apply the second projection and then the first one.  This is not feasible
for two reasons. First, after applying the second projection we lose   symmetry.
Second, our method  removes whole rows in the matrix representation of the tensor in the first projection
(not just specific entries). Hence it is much better to  first remove rows (when the matrix has mostly zeros) and then use the second projection, than to remove rows when the matrix is dense (after the second projection).
\end{remark}

 \section{Compression of tensors: the limits of the substitution method} \label{comprsect}

Consider the product of Grassmannians $\bold G:=G(\aaa',A^*)\times G(\bbb',B^*)\times G(\ccc',C^*)$ with three projections $\pi_i$.
Let 
$\EE=\EE(\aaa',\bbb',\ccc'):=\bigotimes_{i=1}^3\pi_i^*(\cS_i)$ 
be the vector bundle that is the tensor product of the pullbacks of universal subspace bundles $\cS_i$.
  Let   $\p \ra \bold G$ denote the projective bundle with   fiber over   $(A',B',C')$ equal to $Seg(\BP A'\times\BP 
  B' \times\BP C')$, so $\p\subset \BP \EE$.

\begin{definition}
A tensor $T\in A\ot B\ot C$ is {\it  $(\aaa',\bbb',\ccc')$-compression generic (cg)} if there are no subspaces $A'\subset A^*, B'\subset B^*, C'\subset C^*$ of respective dimensions $\aaa',\bbb',\ccc'$ such that $T|_{A'\ot B'\ot C'}=0$, i.e.,
for all $(A',B',C')\in \bold G$, $A'\ot B'\ot C'\not\subset T\upperp$, where 
$T\upperp \subset (A\ot B\ot C)^*$ is the hyperplane annihilating $T$.

Let $X(\aaa',\bbb',\ccc')$ be the set of all tensors that are not $(\aaa',\bbb',\ccc')$-cg.
\end{definition}

\begin{proof}[ Proof of Proposition \ref{prop:fill}]
Let 
$$Y:=\{ (y,[T])\in  \bold G\times \BP (A\ot B\ot C) \mid  \EE_y\subset T\upperp\}.
$$
 Each fiber of the projection $Y\rightarrow \bold G$ is a projective space of dimension $\aaa\bbb\ccc-\aaa'\bbb'\ccc'-1$, so 
$$\dim Y:=(\aaa\bbb\ccc-\aaa'\bbb'\ccc'-1)+(\aaa-\aaa')\aaa'+(\bbb-\bbb')\bbb'+(\ccc-\ccc')\ccc'.$$
On the other hand $X(\aaa',\bbb',\ccc')$ is the projection of $Y$ to $\BP (A\ot B\ot C)$, which proves both claims.
\end{proof}

\bigskip
  
\begin{corollary} \ 
\begin{enumerate}
\item If (\ref{surjdim}) holds then a generic tensor is $(\aaa',\bbb',\ccc')$-cg.
\item If (\ref{surjdim}) does not hold then $\rank \EE^*\leq \dim G(\aaa',A^*)\times G(\bbb',B^*)\times G(\ccc',C^*)$. If the top Chern class of $\EE^*$ is nonzero, then no tensor is   $(\aaa',\bbb',\ccc')$-cg.
\end{enumerate}
\end{corollary}
 
\begin{proof}
The first assertion  is a restatement of Proposition \ref{prop:fill}.

For the second,  notice that $T$   induces a section $\tilde T$ of the vector bundle $\EE^*\ra \bold G$.
The zero locus of $\tilde T$ is $ \{ (A',B',C')\in \bold G \mid  A'\ot B'\ot C'\subset T\upperp\}$. In particular, $\tilde T$ is non-vanishing if and only if $T$ is $(\aaa',\bbb',\ccc')$-cg. If the top Chern class is nonzero, there cannot exist a non-vanishing section.
\end{proof}

\begin{example}\label{ex:eq} Let $\aaa=\bbb=\ccc $ and $\aaa'=\bbb'=\ccc' $.
Then we get non-trivial equations as long as 
$$
\aaa'\geq \lceil \sqrt{3\aaa +\frac 94}-\frac 32\rceil .
$$ 
Thus by this method alone, one potentially gets
border rank equations in $\BC^\aaa\ot\BC^\aaa\ot \BC^\aaa$ up to
$$
3(\aaa - \lceil(\sqrt{3\aaa +\frac 94}-\frac 32)\rceil).
$$
For example, if $\aaa=9$, we may take $\aaa'=4$ and get equations
up to $\s_{15}$.  
\end{example}

\begin{example}
Let $\aaa=\bbb=\ccc=3$. As pointed out by Kileel, the variety $X(2,2,3)$ equals the trifocal variety. 
By the results of Aholt-Oeding \cite{MR3223346} the ideal of this variety is defined by 10 cubics, 81 quintics and 1980 sextics.
\end{example}



In each particular case when there are a finite number of $A'\ot B'\ot C'$ annhilating a generic $T$, 
we may explicitly compute how many different $A'\ot B'\ot C'$ a generic hyperplane
may contain as follows:  
  The Chern polynomial of the dual of the universal bundle is $\sum_{j=0}^k p_{1^j}t^j$, where $p_{1^j}$ is the class corresponding to the Young diagram $1^j$. These classes multiply by the Littlewood-Richardson rule   (in our cases this is the iterated Pieri rule).
  
\begin{example}\label{jasonex}  Let $\aaa=\bbb=\ccc=5$ and $\aaa'=2,\bbb'=1,\ccc'=5$. The bundle $\EE^*$ has rank ten: it is tensor product of a rank $2$ bundle (for $\aaa'$), rank $1$ bundle (for $\bbb'$) and the trivial rank $5$ bundle (for $\ccc'$). 
This example already appeared
in \cite{MR3166392}.  Here  $\bold G= G(2,5)\times \PP^5$ as the last 
Grassmannian degenerates to a point. 
The second Chern class of the tensor product of pull-backs equals:
$$c_2(\pi_1^*(\cS_1)\ot\pi_2^*(\cS_2))=(\yng(1,1),1)+(\yng(1),\yng(1))+(1,\yng(1))^2,$$
where respective Young diagrams represent Schubert classes on $G(2,5)$ and $\PP^5$. 
E.g.~$(1,\yng(1))$ is $G(2,5)$ times a hyperplane in $\PP^5$. 
To compute the top Chern class of $\EE^*$ we need to compute the $5$-th power of the above expression. It will be proportional to the class of a point $(\yng(3,3),\yng(4))$ and we just have to compute the coefficient.

We get the following contributions:
\begin{itemize}
\item $5(\yng(1,1),1)(\yng(1),\yng(1))^4=5\cdot 2=10$.
Indeed, on the second coordinate corresponding to $\PP^5$ we just have to fill, one by one starting from left, the diagram $\yng(4)$. On $G(2,5)$ we must start by filling the two left most entries, by the contribution of $(\yng(1,1),1)$ obtaining:
$\young(xoo,xoo).$ The remaining square (filled with $o$ before) has to be filled with four unit squares. There are two ways to do this: $\young(12,34)$ and $\young(13,24)$.

\item $5\binom 42 (\yng(1,1),1)^2(\yng(1),\yng(1))^2(1,\yng(1))^2=30$, because
there is a unique way here,
\item 
$\binom 52$ corresponding to 
$(\yng(1,1),1)^3(1,\yng(1))^4$.
\end{itemize}
This gives the grand total of $50$.
Hence, in this case the map $Y\rightarrow \p(A\ot B\ot C)$ is surjective, finite with generic fiber of degree $50$.
\end{example}

\bibliographystyle{amsplain}
 
\bibliography{Lmatrix}

\end{document}